\documentclass[a4paper]{IEEEtran}
\usepackage{cite,graphicx,epstopdf,url,amssymb,amsthm,threeparttable,multirow,algorithmic,nameref,paralist}
\usepackage[tight,footnotesize]{subfigure}
\usepackage[ruled,vlined]{algorithm2e}
\usepackage[usenames]{color}
\usepackage[normalem]{ulem}
\usepackage[cmex10]{amsmath}
\IEEEoverridecommandlockouts

\usepackage{diagbox,array}

\theoremstyle{plain}

\newtheorem{theorem}{Theorem}
\newtheorem{lemma}{Lemma}
\newtheorem{cor}{Corollary}

\usepackage{url}

\usepackage[normalem]{ulem}
\newcommand\redout{\bgroup\markoverwith{\textcolor{red}{\rule[.5ex]{2pt}{0.4pt}}}\ULon}

\newcommand{\mb}{\mathbf}
\newcommand{\mc}{\mathcal}

\newcommand{\mbb}{\mathbb}

\newcommand{\A}{\mb{A}}
\newcommand{\B}{\mb{B}}
\newcommand{\D}{\mb{D}}
\newcommand{\X}{\mb{X}}
\newcommand{\Y}{\mb{Y}}
\newcommand{\N}{\mb{N}}
\newcommand{\T}{\mb{T}}
\newcommand{\I}{\mb{I}}

\newcommand{\Sig}{\mb{\Sigma}}

\newcommand{\x}{\mb{x}}
\newcommand{\y}{\mb{y}}
\newcommand{\n}{\mb{n}}
\newcommand{\ba}{\mb{a}}
\newcommand{\bb}{\mb{b}}
\newcommand{\bd}{\mb{d}}

\newcommand{\bbR}{\mbb{R}}
\newcommand{\bbE}{\mbb{E}}
\newcommand{\bbP}{\mbb{P}}

\newcommand{\cS}{\mc{S}}
\newcommand{\cD}{\mc{D}}

\newcommand{\cX}{\mc{X}}
\newcommand{\cE}{\mc{E}}

\newcommand{\wh}{\widehat}

\newcommand{\eps}{\varepsilon}

\newcommand{\al}{\alpha}

\newcommand{\vect}{\mathop{\mathrm{vec}}\nolimits}
\newcommand{\supp}{\mathop{\mathrm{supp}}\nolimits}

\newcommand{\RIP}{\mathsf{RIP}}

\newcommand{\norm}[1]{ \left\| #1 \right\| }
\newcommand{\lr}[1]{ \left\{ #1 \right\} }

\begin{document}

\title{Minimax Lower Bounds for Kronecker-Structured Dictionary Learning}
\pagenumbering{gobble}
\author{
\IEEEauthorblockN{Zahra~Shakeri, Waheed~U.~Bajwa, Anand~D.~Sarwate}

\IEEEauthorblockA{
Dept. of Electrical and Computer Engineering, Rutgers University, Piscataway, New Jersey 08854\\
\texttt{\{zahra.shakeri, waheed.bajwa, anand.sarwate\}@rutgers.edu}
\thanks{The work of the authors was supported in part by the National Science Foundation under awards CCF-1525276 and CCF-1453073, and by the Army Research Office under award W911NF-14-1-0295.
}
}
}
\maketitle
\begin{abstract}
Dictionary learning is the problem of estimating the collection of atomic elements that provide a sparse representation of measured/collected signals or data. This paper finds fundamental limits on the sample complexity of estimating dictionaries for tensor data by proving a lower bound on the minimax risk. This lower bound depends on the dimensions of the tensor and parameters of the generative model. The focus of this paper is on second-order tensor data, with the underlying dictionaries constructed by taking the Kronecker product of two smaller dictionaries and the observed data generated by sparse linear combinations of dictionary atoms observed through white Gaussian noise. In this regard, the paper provides a general lower bound on the minimax risk and also adapts the proof techniques for equivalent results using sparse and Gaussian coefficient models. The reported results suggest that the sample complexity of dictionary learning for tensor data can be significantly lower than that for unstructured data.
\end{abstract}

\section{Introduction}
\label{sec:Introduction}

Dictionary learning has recently received significant attention due to the increased importance of finding sparse representations of signals/data. In dictionary learning, the goal is to construct an overcomplete basis using input signals such that each signal can be described by a small number of atoms (columns)~\cite{kreutz2003dictionary}. Although the existing literature has focused on one-dimensional data, many signals in practice are multi-dimensional and have a tensor structure: examples include 2-dimensional images and 3-dimensional signals produced via magnetic resonance imaging or computed tomography  systems. In traditional dictionary learning techniques, multi-dimensional data are processed after vectorizing of signals. This can result in poor sparse representations as the structure of the data is neglected~\cite{zhang2015denoising}.

In this paper we provide fundamental limits on learning dictionaries for multi-dimensional data with tensor structure: we call such dictionaries \textit{Kronecker-structured} (KS).  Several algorithms have been proposed to learn KS dictionaries~\cite{duan2012k,hawe2013separable,zubair2013tensor,peng2014decomposable,zhang2015denoising,soltani2015tensor} but there has been little work on the theoretical guarantees of such algorithms. The lower bounds we provide on the minimax risk of learning a KS dictionary give a measure to evaluate the performance of the existing algorithms.

In terms of relation to prior work, theoretical insights into classical dictionary learning techniques~\cite{aharon2006uniqueness,agarwal2013learning,agarwal2013exact,arora2013new, schnass2014identifiability,schnass2014local,gribonval2014sparse,jung2014performance,
jung2015minimax} have either focused on achievability of existing algorithms~\cite{aharon2006uniqueness,agarwal2013learning,agarwal2013exact,arora2013new, schnass2014identifiability,schnass2014local,gribonval2014sparse} or lower bounds on minimax risk for one-dimensional data~\cite{jung2014performance,jung2015minimax}. The former works provide sample complexity results for reliable dictionary estimation based on the appropriate minimization criteria~\cite{aharon2006uniqueness,agarwal2013learning,agarwal2013exact,arora2013new,schnass2014identifiability,schnass2014local,gribonval2014sparse}. Specifically, given a probabilistic model for sparse signals and a finite number of samples, a dictionary is recoverable within some distance of the true dictionary as a local minimum of some minimization criterion~\cite{schnass2014identifiability,schnass2014local,gribonval2014sparse}. In contrast, works like Jung et al.~\cite{jung2014performance,jung2015minimax} provide minimax lower bounds for dictionary learning under several coefficient vector distributions and discuss a regime where the bounds are tight for some signal-to-noise (SNR) values. Particularly, for a dictionary $\D\in \bbR^{m\times p}$ and neighborhood radius $r$, they show $N=\mc{O}(r^2mp)$ samples suffices for reliable recovery of the dictionary within its local neighborhood.

While our work is related to that of Jung et al.~\cite{jung2014performance,jung2015minimax}, our main contribution is providing lower bounds for the minimax risk of dictionaries consisting of two coordinate dictionaries that sparsely represent 2-dimensional tensor data. The full version of this work generalizes the results to higher-order tensors~\cite{full_version}. The main approach taken in this regard is the well-understood technique of lower bounding the minimax risk in nonparametric estimation by the maximum probability of error in a carefully constructed multiple hypothesis testing problem~\cite{tsybakov2009introduction,yu1997assouad}. As such, our general approach is similar to the vector case~\cite{jung2015minimax}. Nonetheless, the major challenge in such minimax risk analyses is the construction of appropriate multiple hypotheses, which are fundamentally different in our problem setup due to the Kronecker structure of the true dictionary. In particular, for a dictionary $\D$ consisting of the Kronecker product of two coordinate dictionaries $\A\in\bbR^{m_1\times p_1}$ and $\B\in\bbR^{m_2\times p_2}$, where $m=m_1m_2$ and $p=p_1p_2$, our analysis reduces the sample complexity from $\mc{O}(r^2mp)$ for vectorized data~\cite{jung2015minimax} to $\mc{O}(r^2(m_1p_1+m_2p_2))$. Our results hold even when one of the coordinate dictionaries is not overcomplete (note that both $\A$ and $\B$ cannot be undercomplete, otherwise $\D$ won't be overcomplete). Like previous work~\cite{jung2015minimax}, our analysis is local and our lower bounds depend on the distribution of multidimensional data. Finally, some of our analysis relies on the availability of side information about the signal samples. This suggests that the lower bounds can be improved by deriving them in the absence of such side information.

\textbf{\textit{Notational Convention:}} Underlined bold upper-case, bold upper-case, bold lower-case and lower-case letters are used to denote tensors, matrices, vectors, and scalars, respectively. We write $[K]$ for $\{1,\dots,K\}$. The $k$-th column of a matrix $\X$ is denoted by $\x_k$, while $\X_{\mc{I}}$ denotes the matrix consisting of columns of $\X$ with indices $\mc{I}$, $\sum \X$ denotes the sum of all elements of $\X$, and $\I_d$ denotes the $d\times d$ identity matrix. Also, $\|\mb{v}\|_0$ and $\|\mb{v}\|_2$ denote the $\ell_0$ and  $\ell_2$ norms of the vector $\mb{v}$, respectively, while $\|\X\|_2$ and $\|\X\|_F$ denote the spectral and Frobenius norms of $\X$, respectively.

We write $\X_1 \otimes \X_2$ for the \textit{Kronecker product} of two matrices $\X_1\in \bbR^{m\times n}$ and $\X_2\in \bbR^{p\times q}$: the result is an $mp \times nq$ matrix. Given $\X_1\in \bbR^{m\times n}$ and $\X_2\in \bbR^{p\times n}$, we write $\X_1 \ast \X_2$ for their $mp \times n$ \textit{Khatri-Rao product}~\cite{smilde2005multi}: this is essentially the column-wise Kronecker product of matrices. Given two matrices of the same dimension $\X_1,\X_2 \in \bbR^{m\times n}$, their $m\times n$ \textit{Hadamard product} is denoted by $\X_1 \odot \X_2$, which is the element-wise product of $\X_1$ and $\X_2$. For matrices $\X_1$ and $\X_2$, we define their distance to be $\|\X_1-\X_2\|_F$.
We use $f(\eps) = \mc{O}(g(\eps))$ if $\lim_{\eps \rightarrow 0} f(\eps)/g(\eps)=c < \infty$ for some constant $c$.

\section{Background and Problem Formulation}
\label{sec:SysM}
In the conventional dictionary learning setup, it is assumed that an observation $\y\in \bbR^m$ is generated via a fixed dictionary,
	\begin{align}
	\y = \D \x + \n,
	\end{align}
in which the dictionary $\D\in \bbR^{m\times p}$ is an overcomplete basis ($m<p$) with unit-norm columns, $\x\in \bbR^p$ is the coefficient vector, and $\n \in \bbR^m$ is the underlying noise vector. In contrast to this conventional setup, our focus in this paper is on second-order tensor data. Consider the $2$-dimensional observation $\underline{\Y}\in \bbR^{m_1\times m_2}$. Using any separable transform, $\underline{\Y}$ can be written as
	\begin{align}
	\underline{\Y} = (\mb{T}_1^{-1})^T \underline{\X} \mb{T}_2^{-1},
	\end{align}
where $\underline{\X} \in \bbR^{p_1\times p_2}$ is the matrix of coefficients and $\mb{T}_1\in \bbR^{p_1\times m_1}$ and $\mb{T}_2 \in \bbR^{p_2\times m_2}$ are non-singular matrices transforming the columns and rows of $\underline{\Y}$, respectively. Defining $\A\triangleq (\mb{T}_2^{-1})^T$ and $\B \triangleq (\mb{T}_1^{-1})^T$, we can use a property of Kronecker products~\cite{van2000ubiquitous}, $\vect(\B\underline{\X}\A^T)=(\A\otimes \B)\vect(\underline{\X})$, to get the following expression for $\y \triangleq \vect(\underline{\Y})$:
	\begin{align} \label{eq:model}
	\y = (\A \otimes \B) \x + \n
	\end{align}
for coefficient vector $\x \triangleq \vect(\underline{\X}) \in \bbR^{p}$, and noise vector $\n \in \bbR^{m}$, where $p \triangleq p_1p_2$ and $m \triangleq m_1m_2$. In this work, we assume $N$ independent and identically distributed (i.i.d.) noisy observations $\y_k$ that are generated according to the model in \eqref{eq:model}. Concatenating these observations in $\Y\in \bbR^{m\times N}$, we have
	\begin{align} \label{yk}
	\Y=\D\X+\N,
	\end{align}
where $\D \triangleq \A \otimes \B$ is the unknown KS dictionary, $\X \in \bbR^{p\times N}$ is the coefficient matrix which we initially assume to consist of zero-mean random coefficient vectors with known distribution and covariance $\Sig_x$, and $\N\in \bbR^{m\times N}$ is additive white Gaussian noise (AWGN) with zero mean and variance $\sigma^2$.
	
Our main goal in this paper is to derive conditions under which the dictionary $\D$ can possibly be learned from the noisy observations given in \eqref{yk}. In this regard, we assume the true KS dictionary $\D$ consists of unit norm columns and we carry out local analysis. That is, the true KS dictionary $\D$ is assumed to belong to a neighborhood around a fixed (normalized) reference KS dictionary $\D_0 = \A_0 \otimes \B_0$, i.e., $\|\ba_{0,j}\|_2=1 \ \forall j\in [p_1]$, $\|\bb_{0,j}\|_2=1 \ \forall j\in [p_2]$, and $\D_0 \in \cD$:
	\begin{align}
	\cD &\triangleq  \big\{  \D' \in \bbR^{m\times p} :  \|\bd'_{j}\|_2=1 \ \forall j\in [p], \ \D'= \A'\otimes \B', \nonumber \\
	&\qquad \qquad \qquad  \A' \in \bbR^{m_1\times p_1},
		\B' \in \bbR^{m_2\times p_2} \big\}, \ \text{and}\label{eq:D0}\\
	\D &\in \cX(\D_0,r)  \triangleq  \lr{ \D' \in \cD : \norm{\D' -\D_0}_F^2 < r}, \label{Dclass}
	\end{align}
where the radius $r$ is known. It is worth noting here that, similar to the analysis for vector data~\cite{jung2015minimax}, our analysis is applicable to the global KS dictionary learning problem. Finally, some of our analysis in the following also relies on the notion of the \textit{restricted isometry property} ($\RIP$). Specifically, $\D$ satisfies the $\RIP$ of order $s$ with constant $\delta_s$ if
\begin{align}
\forall~s\text{-sparse } \x, (1-\delta_s)\|\x\|_2^2 \leq \|\D\x\|_2^2 \leq (1+\delta_s)\|\x\|_2^2.
\end{align}

\subsection{Minimax risk analysis} \label{minimax}
We are interested in lower bounding the minimax risk for estimating $\D$ based on observations $\Y$, which is defined as the worst-case mean squared error (MSE) that can be obtained by the best KS dictionary estimator $\wh{\D}(\Y)$. That is,
\begin{align}
	\eps^* = \inf_{\wh{\D}} \sup_{\D\in \cX(\D_0,r)} \bbE_\Y \lr{ \big\|\wh{\D}(\Y)-\D\big\|_F^2}.
	\end{align}
In order to lower bound this minimax risk $\eps^*$, we resort to the multiple hypothesis testing approach taken in the literature on nonparametric estimation~\cite{yu1997assouad,tsybakov2009introduction}. This approach is equivalent to generating a KS dictionary $\D_l$ uniformly at random from a carefully constructed class $\cD_L=\{\D_1, \dots, \D_L\} \subseteq  \cX(\D_0,r), L \geq 2,$ for a given $(\D_0$, $r)$. Observations $\Y=\D_l\X+\N$ in this setting can be interpreted as channel outputs that are fed into an estimator that must decode $\D_l$. A lowerbound on the minimax risk in this setting depends not only on problem parameters such as the number of observations $N$, noise variance $\sigma^2$, dimensions of the true KS dictionary, neighborhood radius $r$, and coefficient distribution, but also on various aspects of the constructed class $\cD_L$~\cite{tsybakov2009introduction}.

To ensure a tight lower bound, we must construct $\cD_L$ such that the distance between any two dictionaries in $\cD_L$ is sufficiently large and the hypothesis testing problem is sufficiently hard, i.e., distinct dictionaries result in similar observations. Specifically, for $l,l' \in [L]$, we desire a construction such that
	\begin{align}
	\forall l \not= l', &\norm{\D_l - \D_{l'} }_F \geq 2\sqrt{2\eps}\quad \text{and}\nonumber\\
	&D_{KL} \left( f_{\D_l}(\Y)|| f_{\D_{l'}}(\Y) \right) \leq \alpha_L,\label{eqn:minimax}
	\end{align}
where $D_{KL} \left( f_{\D_l}(\Y)|| f_{\D_{l'}}(\Y) \right)$ denotes the Kullback-Leibler (KL) divergence between the distributions of observations based on $\D_l \in \cD_L$ and $\D_{l'} \in \cD_L$, while $\eps$ and $\al_L$ are non-negative parameters. Roughly, the minimax risk analysis proceeds as follows. Considering $\wh{\D}(\Y)$ to be an estimator that achieves $\eps^*$, and assuming $\eps^* < \eps$ and $\D_l$ generated uniformly at random from $\cD_L$, we have $\bbP (\wh{l}(\Y)\neq l) = 0$ for the minimum-distance detector $\wh{l}(\Y)$ as long as $\|\wh{\D}(\Y)-\D_l\|_F < \sqrt{2\eps}$. The goal then is to relate $\eps^*$ to $\bbP(\|\wh{\D}(\Y)-\D_l\|_F \geq \sqrt{2\eps})$ and $\bbP (\wh{l}(\Y)\neq l)$ using Fano's inequality~\cite{yu1997assouad}:
	\begin{align}
    \label{eqn:fano}
	(1- \bbP (\wh{l}(\Y)\neq l)) \log_2 L - 1 \leq I(\Y;l),
	\end{align}
where $I(\Y;l)$ denotes the mutual information (MI) between the observations $\Y$ and the dictionary $\D_l$. Notice that the smaller $\alpha_L$ is in \eqref{eqn:minimax}, the smaller $I(\Y;l)$ will be in \eqref{eqn:fano}. Unfortunately, explicitly evaluating $I(\Y;l)$ is a challenging task in our setup because of the underlying distributions. Similar to \cite{jung2015minimax}, we will instead resort to upper bounding $I(\Y;l)$ by assuming access to some side information $\T(\X)$ that will make the observations $\Y$ conditionally multivariate Gaussian (recall that $I(\Y;l) \leq I(\Y;l|\T(\X))$). Our final results will then follow from the fact that any lower bound for $\eps^*$ given the side information $\T(\X)$ will also be a lower bound for the general case~\cite{jung2015minimax}.

\subsection{Coefficient distribution}
The minimax lower bounds in this paper are derived for various coefficient distributions. First, similar to \cite{jung2015minimax}, we consider arbitrary coefficient distributions for which the covariance matrix $\Sig_x$ exists. We then specialize our results for sparse coefficient vectors and, under additional assumptions on the reference dictionary $\D_0$, obtain a tighter lower bound for some signal-to-noise ratio (SNR) regimes, where $\text{SNR} =\frac{\bbE_\x(\|\D\x\|_2^2)}{\bbE_\n(\|\n\|_2^2)}$.

\subsubsection{General coefficients}
The coefficient vector $\x$ in this case is assumed to be a zero-mean random vector with covariance $\Sig_x$. We also assume access to the side information $\T(\X)=\X$ to obtain a lower bound on $\eps^*$ in this setup.
%

\subsubsection{Sparse coefficients}
In this case, we assume $\x$ to be an $s$-sparse vector such that the support of $\x$, denoted by $\supp(\x)$, is uniformly distributed over $\cE=\{\cS\subseteq [p]:|\cS|=s\}$:
	\begin{align} \label{swiss}
	\bbP(\supp(\x)=\cS)=\frac{1}{ \binom{p}{s} }
	\quad \text{for any} \ \cS \in \cE.
\end{align}
Further, we model the nonzero entries of $\x$, i.e., $\x_\cS$, as drawn in an i.i.d. fashion from a distribution with variance $\sigma_a^2$:
	\begin{align} \label{iid}
	\bbE_x\{\x_\cS\x_\cS^T|\cS\}=\sigma_a^2 \I_s.
	\end{align}
Notice that an $\x$ under the assumptions of \eqref{swiss} and \eqref{iid} has
	\begin{align} \label{sig_iid}
	\Sig_x = \frac{s}{p}\sigma_a^2 \I_p.
	\end{align}
Further, it is easy to see in this case that $\text{SNR} =\frac{s\sigma_a^2}{m\sigma^2}$. Finally, the side information assumed in this sparse coefficients setup will either be $\T(\X) = \X$ or $\T(\X) = \supp(\X)$.
%

\section{Lower Bound for General Coefficients}

We now provide our main result for the lower bound for the minimax risk of the KS dictionary learning problem for the case of general (i.i.d.) coefficient vectors.
\begin{theorem}\label{thm_1}
Consider a KS dictionary learning problem with $N$ i.i.d observations generated according to model~\eqref{eq:model} and the true dictionary satisfying \eqref{Dclass} for some $r$ and $\D_0$. Suppose $\Sig_x$ exists for the zero-mean random coefficient vectors. If there exists an estimator with worst-case MSE $\eps^* \leq \frac{2p(1-t)}{8} \min\{1,\frac{r^2}{4p}\}$, then the minimax risk is lower bounded by
	\begin{align} \label{eq:thm_1}
	\eps^* &\geq
		\frac{C_1  r^2\sigma^2}{Np\|\Sig_x\|_2}\left(c_1(p_1(m_1-1)+p_2(m_2-1)) -3\right)
	\end{align}
for any $0<c_1<\frac{t}{8\log 2}$ and $0<t<1$, where $C_1 = \frac{(1-t)p}{32r^2}$.
\end{theorem}

\textit{Outline of Proof:} The idea of the proof, as discussed in section~\ref{minimax}, is that we construct a set of $L$ distinct KS dictionaries that satisfy:
\begin{itemize}
\item $\cD_L =\{\D_1,\dots,\D_L\} \subset \cX(\D_0,r)$

\item Any two distinct dictionaries in $\cD_L$ are separated by a minimum distance in the neighborhood, i.e., for any $l,l' \in [L]$ and some positive $\eps \leq \frac{2p(1-t)}{8} \min\{1,\frac{r^2}{4p}\}$:
	\begin{align} \label{eq:distance}
	\|\D_l-\D_l'\|_F \geq 2\sqrt{2\eps}, \ \text{for} \ l \neq l'.
	\end{align}
\end{itemize}
Notice that if the true dictionary, $\D_l \in \cD_L$, is selected uniformly at random from $\cD_L$ in this case then, given side information $\T(\X)=\X$, the observations $\Y$ follow a multivariate Gaussian distribution and an upper bound on the conditional MI $I(\Y;l|\T(\X))$ can be obtained by using an upper bound for KL-divergence of multivariate Gaussian distributions. This bound depends on parameters $\eps, N, m_1,m_2, p_1,p_2, \Sig_x, s, r$, and $\sigma^2$.

Next, assuming \eqref{eq:distance} holds for $\cD_L$, if there exists an estimator $\wh{\D}(\Y)$ achieving the minimax risk $\eps^*\leq \eps$ and the recovered dictionary $\wh{\D}(\Y)$ satisfies $\|\wh{\D}(\Y)-\D_l\|_F < \sqrt{2\eps}$, the minimum distance detector $\wh{l}(\Y)$ can recover $\D_l$. Consequently, the probability of error $\bbP(\wh{\D}(\Y) \neq \D_l) \leq \bbP(\|\wh{\D}(\Y)-\D_l\|_F \geq \sqrt{2\eps})$ can be used to lower bound the conditional MI using Fano's inequality. The obtained lower bound in our case will only be a function of $L$.

Finally, using the obtained upper and lower bounds for the conditional MI:
	\begin{align} \label{eq:proof_out}
	\eta_2 \leq I(\Y;l|\T(\X))\leq \eta_1,
	\end{align}
a lower bound for the minimax risk $\eps^*$ is attained.

A formal proof of Theorem~\ref{thm_1} relies on the following lemmas whose proofs appear in the full version of this work~\cite{full_version}. Note that since our construction of  $\cD_L$ is more complex than the vector case~\cite[Theorem 1]{jung2015minimax}, it requires a different sequence of lemmas, with the exception of Lemma \ref{lemma_3}, which follows from the vector case.


\begin{lemma} \label{lemma_1}
There exists a set of $L = 2^{c_1(mp) - \frac{1}{2}}$ matrices $\A_l \in \bbR ^{m\times p}$, where elements of $\A_l$ take values $\pm \alpha$ for some $\alpha>0$, such that for $l,l'\in [L]$, $l\ne l'$, any $t>0$ and $c_1 < \frac{1}{2\log 2} \left(\frac{t}{2\alpha^2mp}\right)^2$, the following relation is satisfied:
	\begin{align} \label{eq:lem1}
	\left|\sum (\A_l\odot \A_{l'})\right| \leq t.
	\end{align}
\end{lemma}

\begin{lemma}\label{lemma_2}
Considering the generative model in \eqref{eq:model}, given some $r>0$ and reference dictionary $\D_0$, there exists a set $\cD_L\subseteq \cX(\D_0,r)$ of cardinality $L=2^{c_1((m_1-1)p_1+(m_2-1)p_2)-1}$ such that for any $0<c_1 < \frac{t^2}{8\log 2}$, any $0<t<1$, and any $\eps'>0$ satisfying
	\begin{align}
	\eps' &< \min \left\{ r^2, \frac{r^4}{4p} \right\} \label{eq:eps_r}
	\end{align}
and any $l,l' \in [L]$, with $l \neq l'$, we have
	\begin{align}
	\frac{2p}{r^2}(1-t) \eps' &\leq \|\D_l-\D_{l'}\|_F^2  \leq \frac{8p}{r^2}\eps'.
	\end{align}
Furthermore, considering the general coefficient model for $\X$ and assuming side information $\T(\X)=\X$, we have
	\begin{align} \label{eq:I_1}
	\forall \ l, \ I(\Y;l|\T(\X)) &\leq \frac{4Np\|\Sig_x\|_2}{r^2\sigma^2}\eps'.
	\end{align}
\end{lemma}

\begin{lemma}\label{lemma_3} Consider the generative model in \eqref{eq:model} with minimax risk $\eps^* \leq \eps$ for some $\eps>0$. Assume there exists a finite set $\cD_L \subseteq \cD$ with $L$ dictionaries satisfying
	\begin{align} \label{eq:lem3_cons}
	\|\D_l-\D_{l'}\|_F^2\geq 8\eps
	\end{align}
for $l \neq l'$. Then for any side information $\T(\X)$, we have
	\begin{align}
	I(\Y;l|\T(\X)) \geq \frac{1}{2} \log_2 (L) -1. \label{LB}
	\end{align}
\end{lemma}

\begin{proof}[Proof of Theorem \ref{thm_1}]

According to Lemma~\ref{lemma_2}, for any $\eps'$ satisfying \eqref{eq:eps_r}, there exists a set $\cD_L \subseteq \cX(\D_0,r)$ of cardinality $L=2^{c_1((m_1-1)p_1+(m_2-1)p_2)-1}$ that satisfies \eqref{eq:I_1} for any $c_1<\frac{t}{8\log 2}$ and any $0<t<1$. According to Lemma \ref{lemma_3}, if we set $\frac{2p}{r^2}(1-t)\eps' = 8\eps$, \eqref{eq:lem3_cons} is satisfied for $\cD_L$ and provided there exists an estimator with worst case MSE satisfying $\eps^*\leq \frac{2p(1-t)}{8} \min\{1,\frac{r^2}{4p}\}$, \eqref{LB} holds.
Combining \eqref{eq:I_1} and \eqref{LB} we get
	\begin{align} \label{eq:dovar}
	\frac{1}{2} \log_2 (L) -1 \leq I(\Y;l|\T(\X))
		\leq  \frac{32Np\|\Sig_x\|_2}{c_2r^2\sigma^2}\eps,
	\end{align}
where $c_2 = \frac{2p}{r^2}(1-t)$.
Defining $C_1=\frac{(1-t)p}{32r^2}$, \eqref{eq:dovar} translates into
	\begin{align}
	\eps \geq \frac{C_1 r^2\sigma^2}{Np\|\Sig_x\|_2} \left(c_1(p_1(m_1-1)+p_2(m_2-1)) -3\right). \label{thm_1_eps}
	\end{align}
\end{proof}

\section{Lower Bound for Sparse Coefficients}

We now turn our attention to the case of sparse coefficients and obtain lower bounds for the corresponding minimax risk. We first state a corollary of Theorem \ref{thm_1}, for $\T(\X)=\X$.

\begin{cor} \label{cor_1}
Consider a KS dictionary learning problem with $N$ i.i.d observations according to model~\eqref{eq:model}. Assuming the true dictionary satisfies \eqref{Dclass} for some $r$ and the reference dictionary $\D_0$ satisfies $\RIP(s,\frac{1}{2})$, if the random coefficient vector $\x$ is selected according to \eqref{swiss} and there exists an estimator with worst-case MSE error $\eps^* \leq \frac{2p(1-t)}{8} \min\{1,\frac{r^2}{4p}\}$, the minimax risk is lower bounded by
	\begin{align}\label{eq:cor_1}
	\eps^* &\geq
		\frac{C_1   r^2 \sigma^2 }{N s \sigma_a^2  } \left( c_1 \left(p_1(m_1-1)+p_2(m_2-1) \right) -3 \right)
	\end{align}
for any $0<c_1<\frac{t}{8\log 2}$ and $0<t<1$, where $C_1 = \frac{(1-t)p}{32r^2}$.
\end{cor}
This result is a direct consequence of Theorem \ref{thm_1}, by substituting the covariance matrix of $\X$ given in \eqref{sig_iid} in \eqref{eq:thm_1}.

\subsection{Sparse Gaussian coefficients}
In this section, we make an additional assumption on the coefficient vector generated according to \eqref{swiss} and assume non-zero elements of $\x$ follow a Gaussian distribution.
By additionally assuming the non-zero entries of $\x$ are i.i.d., we can write $\x_\cS$ as
\begin{align} \label{gaussian}
\x_\cS \sim \mc{N}(\mb{0},\sigma_a^2\I_s).
\end{align}
Therefore, given side information $\T(\x)=\supp(\x)$, observations $\y$ follow a multivariate Gaussian distribution. We now provide a theorem for the lower bound attained for this coefficient distribution.

\begin{theorem} \label{thm_2}
Consider a KS dictionary learning problem with $N$ i.i.d observations according to model~\eqref{eq:model}. Assuming the true dictionary satisfies \eqref{Dclass} for some $r$ and the reference coordinate dictionaries $\A_0$ and $\B_0$ satisfy $\RIP(s,\frac{1}{2})$, if the random coefficient vector $\x$ is selected according to \eqref{swiss} and \eqref{gaussian} and there exists an estimator with worst-case MSE error $\eps^* \leq \frac{2p(1-t)}{8} \min\{\frac{1}{s},\frac{r^2}{4p}\}$, then the minimax risk is lower bounded by
	\begin{align} \label{eq:thm_2}
	\eps^* &\geq
		\frac{C_2 r^2\sigma^4 }{Ns^2\sigma_a^4}\left(c_1(p_1(m_1-1)+p_2(m_2-1) )-3\right)
	\end{align}
for any $0<c_1<\frac{t}{8\log 2}$ and $0<t<1$, where $C_2=1.58\times 10^{-5}.\dfrac{p(1-t)}{r^2}$.
\end{theorem}

\begin{figure}
\centering
\includegraphics[width=0.4\textwidth]{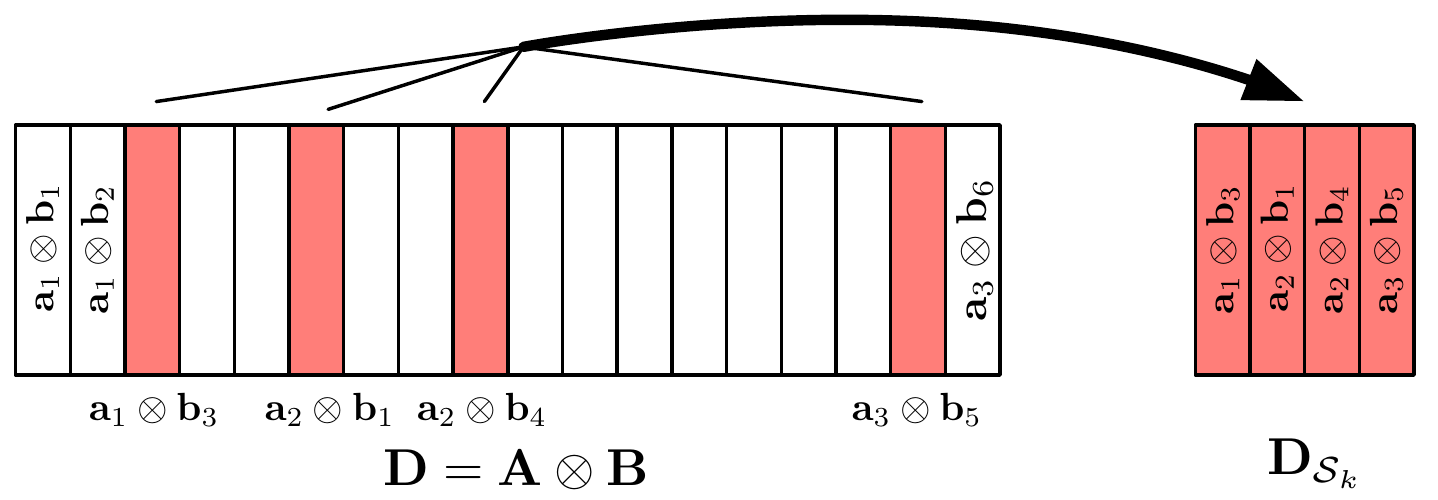}
\caption{An illustration of $\D_{l,\cS_k}$ with $p_1=3$, $p_2=6$ and sparsity $s=4$. Here, $\cS_{ka}=\{1, 2, 2, 3\}$, $\cS_{kb}=\{3, 1, 4, 5\}$, and $\cS_k = \{3,7,10,17\}$.}
\label{fig:1}
\vspace{-1.25\baselineskip}
\end{figure}

\textit{Outline of Proof:}
The constructed dictionary class $\cD_L$ in Theorem \ref{thm_2} is similar to that in Theorem \ref{thm_1}. But the upper bound for the conditional MI, $I(\Y;l|\supp(\X))$, differs from that in Theorem \ref{thm_1}  as the side information is different.

Given the true dictionary $\D_l$ and support $\cS_k$ for the $k$-th coefficient vector $\x_k$, let $\D_{l,\cS_k}$ denote the columns of $\D_l$ corresponding to the non-zeros elements of $\x_k$. In this case, we have
	\begin{align}
	\y_k = \D_{l,\cS_k} \x_{\cS_k} +\n_k,\  k \in [N].
	\end{align}
We can write the subdictionary $\D_{l,\cS_k}$ in terms of the Khatri-Rao product of two smaller matrices:
	\begin{align} \label{eq:KR_p}
	\D_{l,\cS_k} = \A_{l_a,\cS_{ka}} \ast \B_{l_b,\cS_{kb}},
	\end{align}	
where $\cS_{ka} =\{i_k\}_{k=1}^s, i_k \in [p_1]$, and $\cS_{kb} =\{i'_k\}_{k=1}^s, i'_k \in [p_2]$, are multisets with the following relationship with $\cS_k = \{i''_k\}_{k=1}^s, i''_k \in [p]$: $	i''_k = (i_k-1) p_2  + i'_k , \ k \in [s]$.
Note that $\A_{l_a,\cS_{ka}}$ and $\B_{l_b,\cS_{kb}}$ are not submatrices of $\A_{l_a}$ and $\B_{l_b}$, as $\cS_{ka}$ and $\cS_{kb}$ are multisets.
Figure \ref{fig:1} provides a visual illustration of \eqref{eq:KR_p}. Therefore, the observations follow a multivariate Gaussian distribution with zero mean and covariance matrix:
	\begin{align} \label{eq:thm2_sig}
	\Sig_{k,l}  &= \sigma_a^2 (\A_{l_a,S_{ka}} \ast \B_{l_b,\cS_{kb}})(\A_{l_a,S_{ka}} \ast \B_{l_b,\cS_{kb}})^T  + \sigma^2 \I_s
	\end{align}
and we need to obtain an upper bound for the conditional MI using \eqref{eq:thm2_sig}. We state a variation of Lemma \ref{lemma_2} necessary for the proof of Theorem~\ref{thm_2}. The proof of the lemma is again provided in~\cite{full_version}.

\begin{table}
\setlength{\extrarowheight}{5pt}
\renewcommand{\arraystretch}{1.3}
\caption{Order-wise lower bounds on the minimax risk for various coefficient distributions}
\label{table:1}
\centering
\begin{tabular}{|l|c|c|}
\hline
\diagbox[width=2.7cm]{\textbf{Distribution}}{\textbf{Dictionary}}
& Unstructured~\cite{jung2015minimax}  & Kronecker (this paper) \\
\hline
	Sparse  &
	$\dfrac{ r^2 p}{N \text{SNR}}$ &
	$\dfrac{r^2(m_1p_1+m_2p_2)}{N m \text{SNR}}$
	\\[2ex]
\hline
	Gaussian Sparse &
	$\dfrac{r^2 p}{Nm\text{SNR}^2}$ &
	$\dfrac{r^2(m_1p_1+m_2p_2)}{Nm^2\text{SNR}^2}$
	\\[2ex]
\hline
\end{tabular}
\end{table}

%

\begin{lemma}\label{lemma_4} Considering the generative model in \eqref{eq:model}, given some $r>0$ and reference dictionary $\D_0$, there exists a set $\cD_L\subseteq \cX(\D_0,r)$ of cardinality $L=2^{c_1((m_1-1)p_1+(m_2-1)p_2)-1}$ such that for any $0<c_1<\frac{t^2}{8\log 2}$, any $0<t<1$, and any $\eps'>0$ satisfying
	\begin{align} \label{eps:2}
	0<\eps'\leq \min\left\{\frac{r^2}{s},\frac{r^4}{4p}\right\},
	\end{align}
and any $l,l' \in [L]$, with $l \neq l'$, we have
\begin{align} \label{8eps}
\frac{2p}{r^2}(1-t)\eps' \leq \|\D_l-\D_{l'}\|_F^2\leq \frac{8p}{r^2}\eps'.
\end{align}
Furthermore, assuming the reference coordinate dictionaries $\A_0$ and $\B_0$ satisfy  $\RIP(s,\frac{1}{2})$ and the coefficient matrix $\X$ is selected according to \eqref{swiss} and \eqref{gaussian}, considering side information $\T(\X)=\supp(\X)$, we have:
	\begin{align} \label{eq:I_2}
	I(\Y;l|\T(\X))&\leq 7921 \left(\frac{\sigma_a}{\sigma}\right)^4 \frac{Ns^2}{r^2}\eps'.
	\end{align}
\end{lemma}

\begin{proof}[Proof of Theorem \ref{thm_2}]
According to Lemma \ref{lemma_4}, for any $\eps'$ satisfying \eqref{eps:2}, there exists a set $\cD_L \subseteq \cX(\D_0,r)$ of cardinality $L=2^{c_1((m_1-1)p_1+(m_2-1)p_2)-1}$ that satisfies \eqref{eq:I_2} for any $c_1<\frac{t}{8\log 2}$ and any $0<t<1$. Setting $\frac{2p}{r^2}(1-t)\eps'=8\eps$, \eqref{eq:lem3_cons} is satisfied for $\cD_L$ and, provided there exists an estimator with worst case MSE satisfying $\eps^*\leq \frac{2p(1-t)}{8} \min\{\frac{1}{s},\frac{r^2}{4p}\}$, \eqref{LB} holds. Consequently,
	\begin{align} \label{dovar2}
	\frac{1}{2} \log_2 (L) -1 \leq I(\Y;l|\T(\X))
		\leq  \frac{8\times 7921}{c_2} \left(\frac{\sigma_a}{\sigma}\right)^4\frac{Ns^2}{r^2} \eps,
	\end{align}
where $c_2=\frac{2p}{r^2}(1-t)$. Defining $C_2=1.58\times10^{-5}.\frac{p(1-t)}{r^2}$, \eqref{dovar2} can be written as
	\begin{align} \label{thm_2.1_eps}
	\eps \geq C_2\big(\frac{\sigma}{\sigma_a}\big)^4\frac{r^2\left(c_1(p_1(m_1-1)+p_2(m_2-1) )-3\right)}{Ns^2}
		.
	\end{align}
\end{proof}

\section{Discussion and Conclusion}

In this paper we follow an information-theoretic approach to provide lower bounds for the worst-case MSE of KS dictionaries that generate 2-dimensional tensor data. Table \ref{table:1} lists the dependence of the known lower bounds on the minimax rates on various parameters of the dictionary learning problem and the SNR$=\dfrac{s\sigma_a^2}{m\sigma^2}$. Compared to the results in~\cite{jung2015minimax} for the unstructured dictionary learning problem, which are not stated in this form, but can be reduced to this, we are able to decrease the lower bound in all cases by reducing the scaling $\mathcal{O}(pm)$ to $\mathcal{O}(p_1m_1+p_2m_2)$ for KS dictionaries. This is intuitively pleasing since the minimax lower bound has a linear relationship with the number of degrees of freedom of the KS dictionary, which is ($p_1m_1+p_2m_2$), and the square of the neighborhood radius $r^2$. The results also show that the minimax risk decreases with a larger number of samples $N$ and increased SNR.
Notice also that in high SNR regimes, the lower bound in \eqref{eq:cor_1} is tighter, while \eqref{eq:thm_2} results in a tighter lower bound in low SNR regimes.
Our bounds depend on the signal distribution and imply necessary sample complexity scaling $N=\mc{O}(r^2(m_1p_1+m_2p_2)$. Future work includes extending the lower bounds for higher-order tensors and also specifying a learning scheme that achieves these lower bounds.




\end{document}